\documentclass[
  reprint,
  superscriptaddress,
  amsmath,amssymb,
  aps,
  prl,
  nofootinbib,
  nobibnotes,
  longbibliography
]{revtex4-2}

\usepackage[T1]{fontenc}
\usepackage[utf8]{inputenc}
\usepackage{lmodern}
\usepackage{changes}
\usepackage{amsmath, amssymb, amsthm, bm, mathtools}
\usepackage{mathrsfs}
\usepackage{braket}   
\usepackage{bbm}      
\usepackage{graphicx}
\usepackage{xcolor}
\usepackage[caption=false]{subfig}
\usepackage{booktabs}
\usepackage{multirow}
\usepackage{array}
\usepackage{hyperref}
\usepackage[capitalize]{cleveref}
\usepackage{ulem}
\usepackage{xspace}

\crefname{section}{appendix}{appendices}
\crefname{section}{Appendix}{Appendices}

\newtheorem{theorem}{Theorem}

\theoremstyle{definition}

\newtheorem{remark}{Remark}

\newcommand{\prlsec}[1]{
  \par\addvspace{0.8\baselineskip}
  \noindent\textit{#1.---}\hspace{0.35em}\ignorespaces
}

\begin{document}

\title{Generalized Nonlinear Imaginary-Time Evolution}

\author{Chenyu Shi}
\affiliation{Leiden Institute of Advanced Computer Science, Universiteit Leiden, Leiden, The Netherlands}
\affiliation{$\langle aQa^L\rangle$ Applied Quantum Algorithms, Universiteit Leiden, The Netherlands}

\author{Hao Wang}
\email[Corresponding author: ]{h.wang@liacs.leidenuniv.nl}
\affiliation{Leiden Institute of Advanced Computer Science, Universiteit Leiden, Leiden, The Netherlands}
\affiliation{$\langle aQa^L\rangle$ Applied Quantum Algorithms, Universiteit Leiden, The Netherlands}

\author{Jin-Fu Chen}
\email[Corresponding author: ]{jinfuchen@lorentz.leidenuniv.nl}
\affiliation{$\langle aQa^L\rangle$ Applied Quantum Algorithms, Universiteit Leiden, The Netherlands}
\affiliation{Instituut-Lorentz, Universiteit Leiden, P.O. Box 9506, 2300 RA Leiden, The Netherlands}

\begin{abstract}
Imaginary-time evolution (ITE) is a powerful method for ground-state preparation of a given Hamiltonian. The normalized ITE can be viewed as a gradient flow of the energy expectation value with respect to the Fubini--Study metric. In this work, we propose a generalized nonlinear imaginary-time evolution (NITE) for more general quantum state-preparation tasks. We further present a hardware-efficient variational implementation of NITE and reveal its connection to quantum natural gradient descent. NITE is applied to several subroutine tasks, including variance minimization in variational eigensolvers, probe-state preparation in variational quantum sensing, and excited-state preparation using penalty terms. We prove that NITE achieves a local exponential convergence rate under reasonable assumptions. Our results show that NITE outperforms standard gradient descent and can serve as an efficient optimization method for variational tasks beyond ground-state preparation.
\end{abstract}

\maketitle

\prlsec{Introduction}
Imaginary-time evolution (ITE) is a fundamental tool for ground-state preparation and energy minimization
\cite{mottaDeterminingEigenstatesThermal2020,mcardleVariationalAnsatzbasedQuantum2019}. It monotonically decreases the energy expectation value and converges to the ground state for any initial state with nonzero ground-state overlap.  
Recent works studied its properties, applications, and potential implementations on quantum devices~\cite{mottaDeterminingEigenstatesThermal2020, mcardleVariationalAnsatzbasedQuantum2019, stokesQuantumNaturalGradient2020, nishiImplementationQuantumImaginarytime2021, yuanTheoryVariationalQuantum2019, mcmahonEquatingQuantumImaginary2026, gomesAdaptiveVariationalQuantum2021, yeter-aydenizQuantumImaginarytimeEvolution2022, hastingsImprovingQuantumAlgorithms2014}.

From a geometric viewpoint, standard ITE is the gradient flow of the energy expectation value with respect to the Fubini--Study metric \cite{stokesQuantumNaturalGradient2020}. 
This interpretation connects ITE to quantum natural gradient descent (QNGD), showing its superior optimization performance compared with standard gradient descent~\cite{mcardleVariationalAnsatzbasedQuantum2019,stokesQuantumNaturalGradient2020}. However, standard ITE formulations are essentially restricted to expectation value of Hamiltonians, as in typical variational quantum eigensolvers for ground-state preparation \cite{peruzzoVariationalEigenvalueSolver2014a,mccleanTheoryVariationalHybrid2016,tillyVariationalQuantumEigensolver2022b}. 
Recent variational quantum algorithms target more general objectives, including variance minimization for excited-state preparation~\cite{zhangVariationalQuantumEigensolvers2022,hobdayVarianceMinimizationNuclear2025,chenVariationalQuantumEigensolver2023}, quantum Fisher information maximization in quantum metrology~\cite{meyerVariationalToolboxQuantum2021b,beckeyVariationalQuantumAlgorithm2022,maclellanEndtoendVariationalQuantum2024}, and optimization with penalty terms for spectral information~\cite{,higgottVariationalQuantumComputation2019, jonesVariationalQuantumAlgorithms2019}. A natural question is whether ITE can be extended beyond energy minimization to provide a general optimization method for broader variational problems.

To address this question, we begin with the geometric interpretation of standard ITE as the Fubini--Study gradient flow of an expectation-value objective on the space of quantum states \cite{stokesQuantumNaturalGradient2020}.
This observation immediately suggests a generalization: for any differentiable cost function on pure quantum states, one may define the corresponding Fubini--Study gradient flow.
The gradient flow reduces to standard ITE for expectation-value cost functions, whose corresponding unnormalized evolution is governed by a linear differential equation with a state-independent Hermitian generator. For general cost functions, the effective Hermitian generators become state-dependent, making the resulting evolution nonlinear. We call this evolution dynamics \textit{Nonlinear Imaginary-Time Evolution} (NITE).  We further generalize NITE to \textit{Multi-state NITE} to study the low-energy physics, where the cost function is defined over a set of pure states. In this case, NITE is described by a system of coupled nonlinear differential equations for the simultaneous evolution of multiple states. In both extensions, the evolution proceeds along a direction that monotonically decreases the cost function.

In this work, we develop NITE as a general imaginary-time optimization framework beyond energy minimization.
We derive the single-state and multi-state NITE equations, and show that a variational form of NITE yields QNGD. For multi-state cost functions, the metric tensor used in the update step is the sum of the Fubini--Study metric tensors of all states. We prove monotonic decrease of the target cost and 
local exponential convergence for NITE under reasonable assumptions. We apply NITE to variance minimization, Fisher-information maximization, and penalty-based excited-state preparation, and compare it with standard gradient descent. The results show that NITE converges faster and is more robust to random initialization than standard gradient descent. The convergence rate agrees with the prediction of our theorem.

\begin{figure}
\centering
\includegraphics[width=\linewidth, trim=00mm 25mm 00mm 10mm,clip]{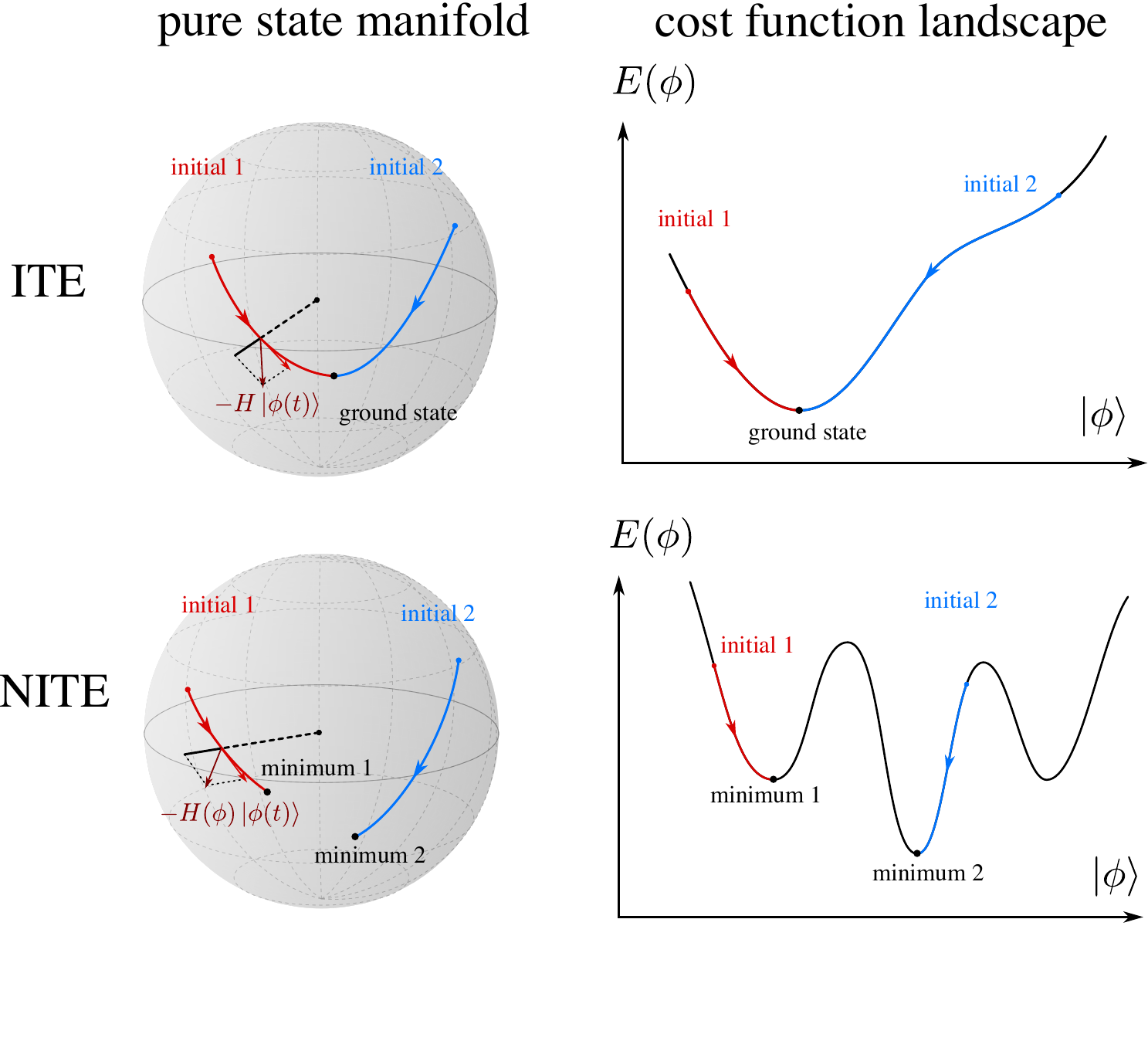}
\caption{Illustrative diagram comparing ITE and NITE. Standard ITE induces an energy-descent flow on the accessible state manifold, driving any initial state with nonzero ground-state overlap toward the ground-state subspace. The update direction is given by the tangent-space projection of $-H\ket{\phi(t)}$. NITE generalizes this picture to state-dependent effective generators associated with general cost functions. Its cost landscape may contain multiple local minima, and the resulting flow can converge to different solutions depending on the initialization. The NITE update direction is also given by the tangent-space projection of $-H(\phi(t))\ket{\phi(t)}$.}
\label{fig:illustrative}
\end{figure}

\prlsec{Imaginary-Time Evolution}
Given an initial pure state $\ket{\phi(0)}$ and a Hamiltonian $H$, ITE is obtained by replacing the real time $t$ in the Schrödinger equation with imaginary time, $t\to -it$:
\begin{equation}
\frac{d}{dt}\ket{\phi(t)}=-H\ket{\phi(t)} .
\label{eq:nonITE}
\end{equation}
Expanding in the energy eigenbasis, we can obtain $\ket{\phi(t)}=\sum_k c_k e^{-E_k t}\ket{E_k}$. If the initial state satisfies $c_0\neq 0$, the normalized evolution converges to the ground state $\ket{E_0}$. The asymptotic convergence is controlled by the spectral gap $E_1-E_0$, with the leading excited-state contribution decaying as $e^{-(E_1-E_0)t}$.

However, \cref{eq:nonITE} does not define a valid evolution of a pure quantum state, since the operator $e^{-Ht}$ is not unitary and therefore does not preserve normalization. To restrict the dynamics to the manifold of normalized pure states, one introduces the normalized form
\begin{equation}
    \frac{d}{dt}\ket{\phi(t)}
    =
    -(H-\braket{H})\ket{\phi(t)},
    \label{eq:norITE}
\end{equation}
where $\braket{H}=\braket{\phi(t)|H|\phi(t)}$ is the energy expectation value. 
Equation~\eqref{eq:norITE} defines an evolution of the normalized state $\ket{\phi(t)}$ that minimizes the energy $\braket{H}$. In particular, Ref. \cite{shiVariationalStudyFermionic2018b} shows that the energy decreases monotonically during the evolution. Consequently, the dynamics always follows a descent direction of the energy functional and converges to the ground state provided that the initial state has nonzero overlap with it.

Though the normalized ITE in \cref{eq:norITE} is nonlinear, the unnormalized ITE in \cref{eq:nonITE} is linear and can be viewed as the counterpart of the linear real-time Schrödinger equation. Hence, we refer to it as linear ITE to distinguish it from the nonlinear generalized framework proposed in this work.

\prlsec{Nonlinear Imaginary-Time Evolution}
A key observation is that ITE in \cref{eq:norITE} is equivalent to a gradient flow with respect to the Fubini--Study metric. The evolution equation for the ket variable is given by
\begin{equation}
    \frac{d}{dt}\ket{\phi(t)}
    = -\nabla^{(\mathrm{FS})} E(\phi),
    \label{eq:ITE}
\end{equation}
where the Fubini--Study gradient $\nabla^{(\mathrm{FS})} E(\phi)$ is obtained by projecting the Wirtinger derivative \cite{kreutz-delgadoComplexGradientOperator2009, koorShortTutorialWirtinger2023a} with respect to the bra variable, i.e., $\frac{\delta E(\phi)}{\delta \bra{\phi}}$, onto the tangent space
\begin{equation}
    \nabla^{(\mathrm{FS})} E(\phi)
    =
    \left(I-\ket{\phi}\bra{\phi}\right)
    \frac{\delta E(\phi)}{\delta \bra{\phi}}.
\end{equation}

When the objective function in \cref{eq:ITE} is the expectation value $E=\bra{\phi}H\ket{\phi}$, the standard ITE \eqref{eq:norITE} is recovered. However, the objective function in \cref{eq:ITE} is not necessarily restricted to the expectation value to define a valid gradient flow. Other functions that define a potential on $\phi$ also define a corresponding gradient flow. 

Denote the gradient with respect to the bra variable by $\ket{g(\phi)}=\nabla E(\phi)=\frac{\delta E(\phi)}{\delta \bra{\phi}}$. One can verify that $\nabla^{(\mathrm{FS})}E(\phi)=P\nabla E(\phi)$, where $P=I-\ket{\phi}\bra{\phi}$ is the projector onto the tangent space at $\ket{\phi}$. Then, we can always define a Hermitian operator $H(\phi)$ determined by $\ket{\phi}$ such that $H(\phi)\ket{\phi}=\ket{g(\phi)}$. Note that the definition of $H(\phi)$ is not unique. To make the definition consistent with linear ITE, we define $H(\phi)$ by the following expression, given the cost function $E(\phi)$ is independent of the global phase of the state $\ket{\phi}$:
\begin{equation}
    H(\phi) = \nabla_{\rho}^{(\mathrm{HS})} \mathcal{E}(\rho)  \Big|_{\rho=\ket{\phi}\bra{\phi}}
    \label{eq:hermitian}
\end{equation}
where $\nabla_{\rho}^{(\mathrm{HS})}$ denotes the gradient with respect to $\rho$ under the Hilbert--Schmidt inner product. Since the cost function $E(\phi)$ is independent of the global phase,
we can always write an equivalent cost function in the density operator form as $\mathcal{E}(\rho)=E(\phi)$, where $\rho=\ket{\phi}\bra{\phi}$.

According to the Riesz representation theorem, $H(\phi)$ defined in \cref{eq:hermitian} is Hermitian.  We also prove that  \cref{eq:hermitian} satisfies $H(\phi)\ket{\phi}=\ket{g(\phi)}$ and reduces to $H$ when the cost function is the expectation value (see \cite{SupplementalMaterial} for details).
From \cref{eq:ITE}, we derive the following equation in the same way by projecting $\ket{g(\phi)}$ onto the tangent space $P\ket{g(\phi)}$:
\begin{equation}
    \frac{d}{dt}\ket{\phi(t)}=-(H(\phi)-\braket{\phi|H(\phi)|\phi})\ket{\phi}\,.
    \label{eq:NITE}
\end{equation}
When the cost function is not an expectation value, the Hamiltonian in \cref{eq:NITE} is state-dependent, rather than a fixed Hermitian operator as in \cref{eq:norITE}. The resulting differential equation is therefore nonlinear in $\phi$. Nevertheless, the optimization can still be implemented through gradient-flow dynamics.

We refer to this method as \textit{Nonlinear Imaginary-Time Evolution} (NITE). Allowing objective functions $E(\phi)$ beyond expectation values substantially extends the applicability of ITE. For instance, one may minimize the variance to suppress fluctuations and drive the state toward an eigenstate, or maximize the quantum Fisher information to target metrologically useful states. Standard ITE is recovered as the special case where the objective function is the energy expectation value. Similar to standard ITE, one can verify that NITE also monotonically decreases $E(\phi)$ from the perspective of gradient flow. 

\prlsec{Multi-state Nonlinear Imaginary-Time Evolution}
The cost function $E(\phi)$ in the previous sections is a functional of a single pure state $\ket{\phi}$.  Here, we generalize it to a functional of multiple pure states $E(\boldsymbol{\phi})$, where $\boldsymbol{\phi}=(\phi_1,\cdots,\phi_K)$. This extension allows one to optimize a cost function defined over a set of quantum states. For example, it provides a natural framework for representing the ground state and a few excited states of a Hamiltonian, thereby describing effective low-energy physics. By constraining the states $\{\ket{\phi_k}\}_{k=1}^K$ to remain orthonormal, one can variationally approximate the lowest $K$ eigenstates.

We define the evolution of the state $\ket{\phi_k}$ as the gradient flow under the Fubini--Study metric as
\begin{equation}
\frac{d}{dt}\ket{\phi_k(t)}=-\nabla_{\bra{\phi_k}}^{(\mathrm{FS})}E(\boldsymbol{\phi})\,.
\label{eq:NITEmulti}
\end{equation}
Instead of the NITE of a single pure state, the gradient flow defines a system of NITE equations of multiple pure states
\begin{equation}
\frac{d}{dt}\ket{\phi_k(t)}=
-(H_k(\boldsymbol{\phi})-\braket{\phi_k|H_k(\boldsymbol{\phi})|\phi_k})\ket{\phi_k}\,.
\label{eq:NNITEmulti}
\end{equation}

The Hermitian operator $H_k(\boldsymbol{\phi})$ is defined in a similar way as in \cref{eq:hermitian}:
\begin{equation}
    H_k(\boldsymbol{\phi}) = \nabla_{\rho_k}^{(\mathrm{HS})} \mathcal{E}(\boldsymbol{\rho})\Big|_{\boldsymbol{\rho}=\boldsymbol{\ket{\phi}}\boldsymbol{\bra{\phi}}}\,,
    \label{eq:hermitianvec}
\end{equation}
where we use the vector notation $\boldsymbol{\ket{\phi}}=(\ket{\phi_1},\ldots,\ket{\phi_K})$, $\boldsymbol{\bra{\phi}}=(\bra{\phi_1},\ldots,\bra{\phi_K})$, $\ket{\boldsymbol{\phi}}\bra{\boldsymbol{\phi}}={(\ket{\phi_1}\bra{\phi_1},\ldots,\ket{\phi_K}\bra{\phi_K})}$ and $\boldsymbol{\rho}=(\rho_1,\ldots,\rho_K)$. Since each $H_k$ is a functional of $\boldsymbol{\phi}$, the system of $K$ differential equations contains coupling terms. One can verify that the multi-state NITE defined in \cref{eq:NNITEmulti} also evolves to monotonically decrease the objective function $E(\boldsymbol{\phi})$.

\prlsec{Convergence Analysis}
It is known that standard ITE exhibits exponential convergence when the initial state has nonzero overlap with the ground state. The objective gap satisfies $E(t)-E_0=\mathcal{O}(e^{-2\Delta t})$, where $\Delta=E_1-E_0$ is the spectral gap of the given Hamiltonian \cite{hartungConvergenceEfficiencyProof2025, benavides-riverosAugmentingImaginaryTimeEvolution2026}. We present the local convergence theorem for NITE below.
\begin{theorem}[Exponential convergence to local minima]
Let $\mathcal{M}_{K}=(\mathbb{CP}^{d-1})^K$ be equipped with the product
Fubini--Study metric, and let $E:\mathcal{M}_{K}\to\mathbb{R}$ be $C^2$.
Suppose that $\boldsymbol{\phi}_{*}$ is a non-degenerate local minimum of
$E$ with $\lambda_*=\lambda_{\min}\left(\operatorname{Hess}_{\mathrm{FS}}E(\boldsymbol{\phi}_{*})\right)>0 $. Then, for any $\mu\in(0,2\lambda_*)$, there exists a neighborhood
$V$ of $\boldsymbol{\phi}_{*}$ such that, for every initial point
$\boldsymbol{\phi}(0)\in V$, the evolution of NITE satisfies, for all $t\ge 0$,
\begin{equation*}
        E(\boldsymbol{\phi}(t))-E(\boldsymbol{\phi}_{*})
    \le
    e^{-\mu t}
    \left(
        E(\boldsymbol{\phi}(0))-E(\boldsymbol{\phi}_{*})
    \right).
\end{equation*}
\vspace{-7mm}
\label{th:linear_converge}
\end{theorem}
The proof of the theorem is provided in \cite{SupplementalMaterial}. This theorem shows that, for a twice continuously differentiable cost function, NITE converges exponentially to a non-degenerate local minimum whose Fubini--Study Hessian is strictly positive definite. Note that the convergence rate of standard ITE can be viewed as a special case of the NITE result. For standard ITE, the positive-definiteness requirement in the theorem reduces to the spectral-gap condition $\Delta>0$, where the minimum eigenvalue $\lambda_*$ of the Hessian is exactly the energy gap $\Delta$. Details of the remark are provided in \cite{SupplementalMaterial}. This shows that the convergence rate factor in the theorem is consistent with the standard ITE case.

\prlsec{Generalization of QNGD}
Variational quantum circuits provide a hardware-efficient approach for near-term quantum devices, which parameterizes a submanifold of $\mathbb{CP}^{n-1}$ by a set of parameters $\theta$. QNGD \cite{stokesQuantumNaturalGradient2020} is a numerical optimization method for adjusting the parameters $\theta$ in variational quantum circuits. QNGD uses the Fubini--Study metric in the pure-state space instead of the Euclidean geometry in parameter space, leading to the following update rule:
\begin{equation}
    \Delta \theta = -\eta G_{\theta}^{-1}\nabla f(\phi(\theta))\,,
    \label{eq:qngd}
\end{equation}
where $[G_{\theta}]_{ij}=\operatorname{Re}[\braket{\partial_{\theta_i}\phi|\partial_{\theta_j}\phi}-\braket{\partial_{\theta_i}\phi|\phi}\braket{\phi|\partial_{\theta_j}\phi}]$ is the Fubini--Study metric tensor and $\eta$ is the learning rate.

Previous works have observed that, compared to standard gradient descent, QNGD can accelerate convergence \cite{stokesQuantumNaturalGradient2020, koczorQuantumNaturalGradient2022} and avoid local minima \cite{wierichsAvoidingLocalMinima2020}. It has also been proven that QNGD is equivalent to the variational form of ITE for ground-state preparation \cite{stokesQuantumNaturalGradient2020}. Although the definition of QNGD in \cite{stokesQuantumNaturalGradient2020} does not require the cost function to be an expectation value, its established equivalence to variational ITE has so far mainly been studied in energy minimization.
We prove in \cite{SupplementalMaterial} that QNGD with general cost functions is the variational form of the single-state NITE defined in \cref{eq:NITE}.

Moreover, for the multi-state NITE defined in \cref{eq:NNITEmulti}, we prove in \cite{SupplementalMaterial} that the generalized QNGD follows the update rule
\begin{equation}
\Delta \theta = -\eta \bigg(\sum_k G_k(\theta)\bigg)^{-1}\nabla_{\theta} f(\phi(\theta)),
\label{eq:qngdmulti}
\end{equation}
where $G_k(\theta)$ is the Fubini--Study metric tensor of each pure state $\ket{\phi_k}$ in $E(\boldsymbol{\phi})$. This update rule can be viewed as a special case of WA-QNGD \cite{shiWeightedApproximateQuantum2026}, where the system consists of product states.

NITE provides a theoretical explanation for QNGD with general cost functions from the perspective of state evolution. Conversely, QNGD in \cref{eq:qngd} provides a variational formulation of NITE that is better suited to near-term quantum devices. In the following examples, we examine the performance of NITE based on this variational formulation in several instances to demonstrate its potential applications.

\prlsec{Variational Quantum Eigensolvers}
Eigenstate preparation is an important problem in quantum chemistry and many-body nuclear systems \cite{zhangVariationalQuantumEigensolvers2022, hobdayVarianceMinimizationNuclear2025}. In related methods \cite{zhangVariationalQuantumEigensolvers2022, hobdayVarianceMinimizationNuclear2025, chenVariationalQuantumEigensolver2023}, a subroutine of variational quantum eigensolvers requires minimizing the variance
\begin{equation}
    E(\phi)=\braket{\phi|H^2|\phi}-(\braket{\phi|H|\phi})^2\,.
    \label{eq:variance}
\end{equation}
Here, we consider this cost function as an example to examine NITE and compare it with standard gradient descent. We also implement NITE in its variational form using a parameterized quantum circuit according to \cref{eq:qngd}. Following the problem setting in \cite{zhangVariationalQuantumEigensolvers2022}, we evaluate our method on a 6-qubit $H_3^{+}$ Hamiltonian from quantum chemistry. The specific form of the Hamiltonian is taken from the PennyLane Molecular Library \cite{Utkarsh2023Chemistry}. The ansatz uses an 8-layer EfficientSU2 circuit. Since our goal is to assess the performance of NITE as an optimization method, we focus primarily on its convergence speed from a random initial point to an eigenstate.

\begin{figure}
\centering
\includegraphics[width=0.99\linewidth, trim=1mm 1mm 0mm 0mm, clip]{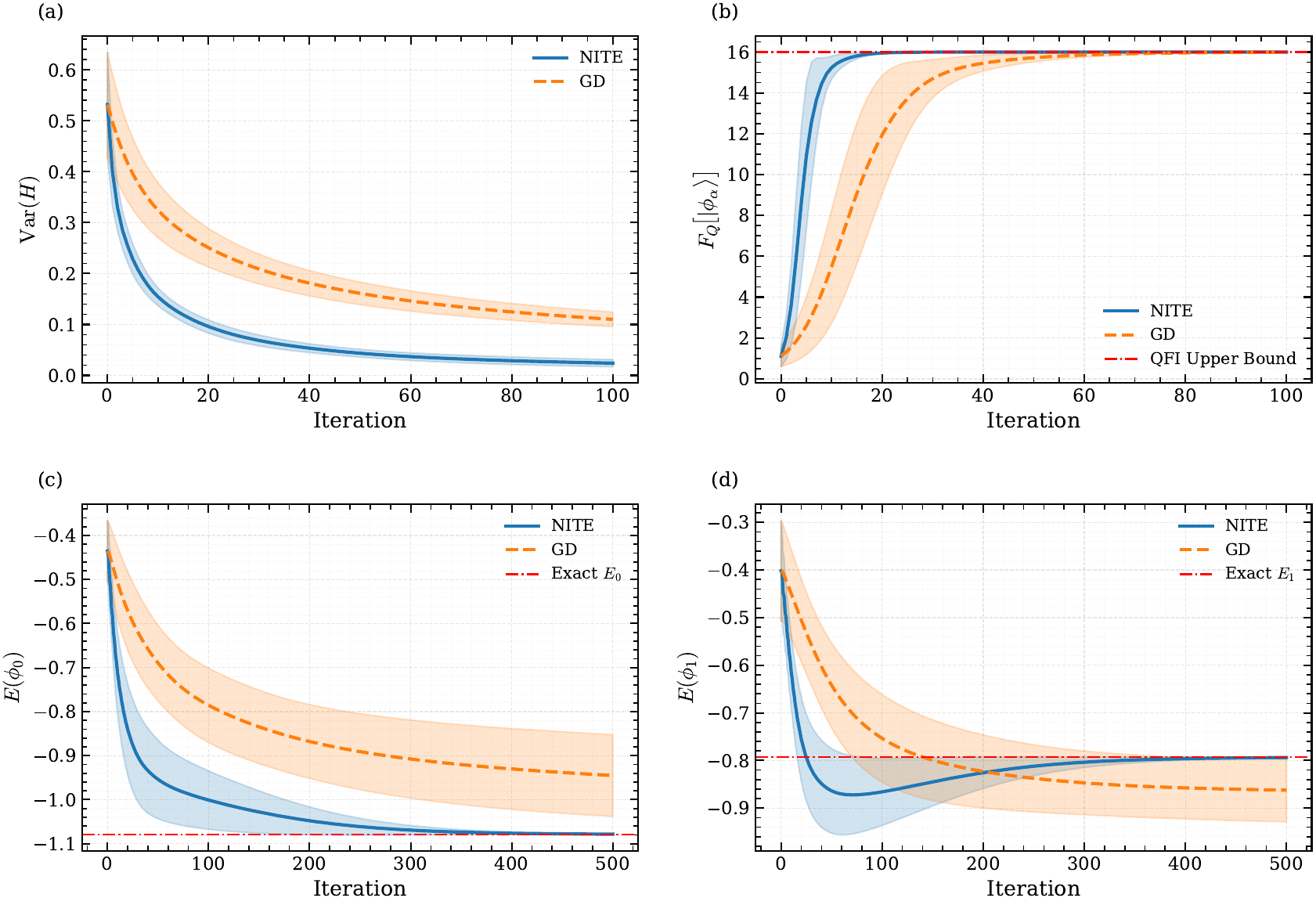}
\vspace{-5mm}
\caption{
Comparison of NITE and standard gradient descent learning curves in three subroutine tasks. 
(a) Variance minimization.
(b) QFI maximization, with the red dash-dotted line indicating the theoretical QFI upper bound.
(c,d) Simultaneous preparation of the ground state and first excited state, where the two panels show the energies of the prepared states and the red dash-dotted lines indicate the corresponding exact eigen energies.
Solid blue curves denote NITE, dashed orange curves denote GD, and shaded regions show one standard deviation over $10$ independent runs.
}
\label{fig:variance}
\end{figure}

The learning curves are shown in \cref{fig:variance}. The numerical simulations are conducted 10 times with different initial points. As we can see, the evolution defined by NITE converges faster to a minimum than standard gradient descent, indicating that NITE has advantages for the variance minimization subroutine. Hence, using NITE in related tasks could benefit the corresponding algorithms.

Due to the limited expressivity of the ansatz, the learning curve of NITE still does not converge to exactly zero, although it performs better than gradient descent. To avoid the limitation of the ansatz, we conduct additional numerical experiments in \cite{SupplementalMaterial} using full-space parameterization to ensure that the solution is contained within the search space. 

\prlsec{Probe State Preparation for Quantum Metrology}
Variational quantum sensing is a technique used for quantum metrology and parameter estimation \cite{meyerVariationalToolboxQuantum2021b,beckeyVariationalQuantumAlgorithm2022,maclellanEndtoendVariationalQuantum2024}. A subroutine in variational quantum sensing is preparing the probe state $\phi$ that maximizes the Fisher information with the parameter $\alpha$ to be estimated. Here we consider this subroutine with the optimal measurement strategy. Namely, the objective to be maximized is the quantum Fisher information $F_Q$ of $\alpha$, which is an upper bound on the Fisher information over all possible measurements.

We consider a single-parameter phase-sensing problem, where the unknown parameter $\alpha$ is encoded into a probe state $\ket{\phi}$ through a unitary evolution generated by a Hermitian operator $H$. For the encoded pure state $\ket{\phi_\alpha}$, the quantum Fisher information (QFI) is defined as
\begin{equation}
\!\!\!\!\!\!\!F_Q[\ket{\phi_\alpha}]
=
4\left(
\langle\partial_{\alpha}\phi_{\alpha}|\partial_{\alpha}\phi_{\alpha}\rangle
-
\langle\phi_{\alpha}|\partial_{\alpha}\phi_{\alpha}\rangle
\langle\partial_{\alpha}\phi_{\alpha}|\phi_{\alpha}\rangle
\right).
\end{equation}
For unitary parameter encoding generated by $\ket{\phi_\alpha }=e^{-i\alpha H}\ket{\phi}$, this reduces to $    F_Q[\ket{\phi_\alpha}]
    =
    4[
    \braket{\phi|H^2|\phi}
    -
    \braket{\phi|H|\phi}^2
    ]$. Therefore, optimizing the probe state $\ket{\phi}$ to maximize the QFI is equivalent to minimizing
\begin{equation}
    E(\phi)
    =
    -F_Q[\ket{\phi_\alpha}].
\end{equation}

In our numerical example, we follow the single-parameter phase-sensing setting of Ref.~\cite{maclellanEndtoendVariationalQuantum2024}, taking $H=\sum_{i=1}^n Z_i$ as the collective generator of a uniform $z$-axis phase rotation. The resulting QFI quantifies the sensitivity of the variational probe state to this collective phase encoding.
\label{eq:sensing}
We compare NITE with gradient descent in a $4$-qubit instance of this subroutine task, and the QFI upper bound in this case is $16$. The NITE is also implemented in its variational form using a parameterized quantum circuit as \cref{eq:qngd}. The ansatz uses a $2$-layer EffcientSU2 circuit. The numerics are conducted 10 times with different initializations. As shown in \cref{fig:variance}, NITE still exhibits faster convergence than gradient descent. 
Additional numerical experiments using full-space parameterization are provided in \cite{SupplementalMaterial} to avoid the expressibility limitation of the ansatz.

\prlsec{Preparing Excited State with Penalty Terms}
Penalty terms are a standard tool for targeting excited states \cite{higgottVariationalQuantumComputation2019,jonesVariationalQuantumAlgorithms2019}. Here we consider the cost function for preparing the first excited state and the ground state simultaneously:
\begin{equation}
\!\!\!\!\! E(\phi_0,\phi_1)\!=\!\lambda_0\!\braket{\phi_0|H|\phi_0} \!+\! \lambda_1 \!\braket{\phi_1|H|\phi_1}+\!\lambda_p\! \left| \braket{\phi_0|\phi_1} \right|^2.
\label{eq:pfcs}
\end{equation}

According to \cref{eq:NITEmulti} and \cref{eq:NNITEmulti}, the above cost function leads to two NITE dynamics for the states $\ket{\phi_0}$ and $\ket{\phi_1}$, respectively. By choosing appropriate $\lambda_0$, $\lambda_1$, and $\lambda_p$, the two states are expected to converge to the ground state and the first excited state. The cost function can be generalized to the first $K$ eigenstates of the Hamiltonian $H$
\begin{equation}
E(\boldsymbol{\phi})=\sum_{k=0}^{K-1}\lambda_{k}\bra{\phi_{k}}H\ket{\phi_{k}}+\!\!\!\!\!\!\sum_{0\le i<j\le K-1}\!\!\!\!\lambda_{ij}\bigl|\braket{\phi_{i}|\phi_{j}}\bigr|^{2},
\label{eq:pfcs-general}
\end{equation}
where the coefficients satisfy $\lambda_k>0$ and $\lambda_{ij}>0$. Assuming that $H$ has a non-degenerate spectrum, there exist choices of the coefficients $\lambda_k$ and $\lambda_{ij}$ such that the global minimum of \cref{eq:pfcs-general} is achieved by the first $K$ eigenstates of $H$.

Here, we consider the cost function \cref{eq:pfcs} as an instance to investigate the performance of NITE and compare it with standard gradient descent for preparing the ground and first excited states simultaneously. We also implement NITE in its variational form using a parameterized quantum circuit with the update rule described in \cref{eq:qngdmulti}. We evaluate our approach on a 4-qubit $H_2$ Hamiltonian from quantum chemistry. The explicit form of the Hamiltonian is also obtained from the PennyLane Molecular Library.

The energy curves of $\ket{\phi_0}$ and $\ket{\phi_1}$ are shown in \cref{fig:variance}. The numerical results are averaged over 10 independent runs. During the training process, the NITE energy curves of $\ket{\phi_0}$ and $\ket{\phi_1}$ converge to the ground-state energy and the first-excited-state energy, respectively. In contrast, gradient descent fails to converge to the desired states. This result indicates that NITE not only achieves a faster convergence speed, but also improves the quality of convergence, successfully identifying the ground state and the first excited state. Moreover, the learning curves of NITE exhibit smaller variance compared with those of gradient descent, suggesting that NITE is more robust to initialization.
To avoid the influence of limited expressility of the ansatz, we also conduct the numerics with full-space parameterization in \cite{SupplementalMaterial}.

\prlsec{Conclusion and Outlook}
In this work, we introduced nonlinear imaginary-time evolution (NITE) as a generalization of standard imaginary-time evolution to objectives beyond the expectation value of a fixed Hamiltonian. The resulting effective Hermitian generator is state-dependent, leading to a nonlinear differential equation. We also formulated a multi-state extension, where the objective is jointly determined by several states and the dynamics is described by coupled nonlinear differential equations. This provides a variational route to target a low-energy subspace containing the ground state and nearby excited states.

NITE is not restricted to specific parameterization methods, such as variational quantum circuits, and can be extended to other variational approaches, including matrix product states~\cite{schollwockDensitymatrixRenormalizationGroup2011c}, non-Gaussian variational states~\cite{shiVariationalStudyFermionic2018b}, and neural quantum states~\cite{carleoSolvingQuantumManybody2017a}. 
Moreover, the same framework could also be generalized to finite-temperature variational manifolds \cite{shiVariationalApproachManyBody2020, luVariationalNeuralTensor2025}, where the free-energy functional contains an entropic contribution and therefore defines a nonlinear cost function on the purified representation of thermal states.

We also emphasize that NITE is not limited to physical optimization problems. In variational-ansatz-based quantum machine learning \cite{mitaraiQuantumCircuitLearning2018}, the cost function is often not a single expectation value, but a nonlinear composite of expectation values, such as least-square loss or cross-entropy loss in variational quantum classifiers \cite{schuldCircuitcentricQuantumClassifiers2020, mariTransferLearningHybrid2020}. In such settings, NITE may provide a useful geometry-aware optimization method for improving the training process.

\prlsec{Acknowledgements}
This project was funded by the European Union (ERC CoG, BeMAIQuantum, 101124342). This project was also co-funded by the Dutch National Growth Fund (NGF), as part of the Quantum Delta NL programme. J.F.C. also
acknowledges the support received from the
European Union’s Horizon Europe research and
innovation programme through the ERC StG
FINE-TEA-SQUAD (Grant No. 101040729).
The views and opinions expressed here are
solely those of the authors and do not necessarily
reflect those of the funding institutions. None of
the funding institutions can be held responsible
for them.

\bibliographystyle{unsrt}  
\bibliography{NITE,refadded}  

\begin{thebibliography}{10}

\bibitem{mottaDeterminingEigenstatesThermal2020}
Mario Motta, Chong Sun, Adrian T.~K. Tan, Matthew~J. O’Rourke, Erika Ye, Austin~J. Minnich, Fernando G. S.~L. Brandão, and Garnet Kin-Lic Chan.
\newblock Determining eigenstates and thermal states on a quantum computer using quantum imaginary time evolution.
\newblock {\em Nature Physics}, 16(2):205--210, February 2020.

\bibitem{mcardleVariationalAnsatzbasedQuantum2019}
Sam McArdle, Tyson Jones, Suguru Endo, Ying Li, Simon~C. Benjamin, and Xiao Yuan.
\newblock Variational ansatz-based quantum simulation of imaginary time evolution.
\newblock {\em npj Quantum Information}, 5(1):75, September 2019.

\bibitem{stokesQuantumNaturalGradient2020}
James Stokes, Josh Izaac, Nathan Killoran, and Giuseppe Carleo.
\newblock Quantum {Natural} {Gradient}.
\newblock {\em Quantum}, 4:269, May 2020.

\bibitem{nishiImplementationQuantumImaginarytime2021}
Hirofumi Nishi, Taichi Kosugi, and Yu-ichiro Matsushita.
\newblock Implementation of quantum imaginary-time evolution method on {NISQ} devices by introducing nonlocal approximation.
\newblock {\em npj Quantum Information}, 7(1):85, June 2021.

\bibitem{yuanTheoryVariationalQuantum2019}
Xiao Yuan, Suguru Endo, Qi~Zhao, Ying Li, and Simon~C. Benjamin.
\newblock Theory of variational quantum simulation.
\newblock {\em Quantum}, 3:191, October 2019.

\bibitem{mcmahonEquatingQuantumImaginary2026}
Nathan~A. McMahon, Mahum Pervez, and Christian Arenz.
\newblock Equating quantum imaginary time evolution, {Riemannian} gradient flows, and stochastic implementations.
\newblock {\em Physical Review Research}, 8(2):023024, April 2026.

\bibitem{gomesAdaptiveVariationalQuantum2021}
Niladri Gomes, Anirban Mukherjee, Feng Zhang, Thomas Iadecola, Cai‐Zhuang Wang, Kai‐Ming Ho, Peter~P. Orth, and Yong‐Xin Yao.
\newblock Adaptive {Variational} {Quantum} {Imaginary} {Time} {Evolution} {Approach} for {Ground} {State} {Preparation}.
\newblock {\em Advanced Quantum Technologies}, 4(12):2100114, December 2021.

\bibitem{yeter-aydenizQuantumImaginarytimeEvolution2022}
Kübra Yeter-Aydeniz, Eleftherios Moschandreou, and George Siopsis.
\newblock Quantum imaginary-time evolution algorithm for quantum field theories with continuous variables.
\newblock {\em Physical Review A}, 105(1):012412, January 2022.

\bibitem{hastingsImprovingQuantumAlgorithms2014}
M.~B. Hastings, D.~Wecker, B.~Bauer, and M.~Troyer.
\newblock Improving {Quantum} {Algorithms} for {Quantum} {Chemistry}.
\newblock 2014.
\newblock Version Number: 2.

\bibitem{peruzzoVariationalEigenvalueSolver2014a}
Alberto Peruzzo, Jarrod McClean, Peter Shadbolt, Man-Hong Yung, Xiao-Qi Zhou, Peter~J. Love, Alán Aspuru-Guzik, and Jeremy~L. O’Brien.
\newblock A variational eigenvalue solver on a photonic quantum processor.
\newblock {\em Nature Communications}, 5(1):4213, July 2014.

\bibitem{mccleanTheoryVariationalHybrid2016}
Jarrod~R McClean, Jonathan Romero, Ryan Babbush, and Alán Aspuru-Guzik.
\newblock The theory of variational hybrid quantum-classical algorithms.
\newblock {\em New Journal of Physics}, 18(2):023023, February 2016.

\bibitem{tillyVariationalQuantumEigensolver2022b}
Jules Tilly, Hongxiang Chen, Shuxiang Cao, Dario Picozzi, Kanav Setia, Ying Li, Edward Grant, Leonard Wossnig, Ivan Rungger, George~H. Booth, and Jonathan Tennyson.
\newblock The {Variational} {Quantum} {Eigensolver}: {A} review of methods and best practices.
\newblock {\em Physics Reports}, 986:1--128, November 2022.

\bibitem{zhangVariationalQuantumEigensolvers2022}
Dan-Bo Zhang, Bin-Lin Chen, Zhan-Hao Yuan, and Tao Yin.
\newblock Variational quantum eigensolvers by variance minimization.
\newblock {\em Chinese Physics B}, 31(12):120301, November 2022.

\bibitem{hobdayVarianceMinimizationNuclear2025}
I.~Hobday, P.~D. Stevenson, and J.~Benstead.
\newblock Variance minimization for nuclear structure on a quantum computer.
\newblock {\em Physical Review C}, 111(6):064321, June 2025.

\bibitem{chenVariationalQuantumEigensolver2023}
Bin-Lin Chen and Dan-Bo Zhang.
\newblock Variational {Quantum} {Eigensolver} with {Mutual} {Variance}-{Hamiltonian} {Optimization}.
\newblock {\em Chinese Physics Letters}, 40(1):010303, January 2023.

\bibitem{meyerVariationalToolboxQuantum2021b}
Johannes~Jakob Meyer, Johannes Borregaard, and Jens Eisert.
\newblock A variational toolbox for quantum multi-parameter estimation.
\newblock {\em npj Quantum Information}, 7(1):89, June 2021.

\bibitem{beckeyVariationalQuantumAlgorithm2022}
Jacob~L. Beckey, M.~Cerezo, Akira Sone, and Patrick~J. Coles.
\newblock Variational quantum algorithm for estimating the quantum {Fisher} information.
\newblock {\em Physical Review Research}, 4(1):013083, February 2022.

\bibitem{maclellanEndtoendVariationalQuantum2024}
Benjamin MacLellan, Piotr Roztocki, Stefanie Czischek, and Roger~G. Melko.
\newblock End-to-end variational quantum sensing.
\newblock {\em npj Quantum Information}, 10(1):118, November 2024.

\bibitem{higgottVariationalQuantumComputation2019}
Oscar Higgott, Daochen Wang, and Stephen Brierley.
\newblock Variational {Quantum} {Computation} of {Excited} {States}.
\newblock {\em Quantum}, 3:156, July 2019.

\bibitem{jonesVariationalQuantumAlgorithms2019}
Tyson Jones, Suguru Endo, Sam McArdle, Xiao Yuan, and Simon~C. Benjamin.
\newblock Variational quantum algorithms for discovering {Hamiltonian} spectra.
\newblock {\em Physical Review A}, 99(6):062304, June 2019.

\bibitem{shiVariationalStudyFermionic2018b}
Tao Shi, Eugene Demler, and J.~Ignacio~Cirac.
\newblock Variational study of fermionic and bosonic systems with non-{Gaussian} states: {Theory} and applications.
\newblock {\em Annals of Physics}, 390:245--302, March 2018.

\bibitem{kreutz-delgadoComplexGradientOperator2009}
Ken Kreutz-Delgado.
\newblock The {Complex} {Gradient} {Operator} and the {CR}-{Calculus}, 2009.
\newblock Version Number: 1.

\bibitem{koorShortTutorialWirtinger2023a}
Kelvin Koor, Yixian Qiu, Leong~Chuan Kwek, and Patrick Rebentrost.
\newblock A short tutorial on {Wirtinger} {Calculus} with applications in quantum information, 2023.
\newblock Version Number: 1.

\bibitem{SupplementalMaterial}
Supplemental {Material}.

\bibitem{hartungConvergenceEfficiencyProof2025}
Tobias Hartung and Karl Jansen.
\newblock Convergence and efficiency proof of quantum imaginary time evolution for bounded order systems, 2025.
\newblock Version Number: 2.

\bibitem{benavides-riverosAugmentingImaginaryTimeEvolution2026}
Carlos~L. Benavides-Riveros, Prachi Sharma, and Fedor Šimkovic.
\newblock Augmenting {Imaginary}-{Time} {Evolution} with {Local} {Geometric} {Information}, 2026.
\newblock Version Number: 1.

\bibitem{koczorQuantumNaturalGradient2022}
Bálint Koczor and Simon~C. Benjamin.
\newblock Quantum natural gradient generalized to noisy and nonunitary circuits.
\newblock {\em Physical Review A}, 106(6):062416, December 2022.

\bibitem{wierichsAvoidingLocalMinima2020}
David Wierichs, Christian Gogolin, and Michael Kastoryano.
\newblock Avoiding local minima in variational quantum eigensolvers with the natural gradient optimizer.
\newblock {\em Physical Review Research}, 2(4):043246, November 2020.

\bibitem{shiWeightedApproximateQuantum2026}
Chenyu Shi, Vedran Dunjko, and Hao Wang.
\newblock Weighted approximate quantum natural gradient for variational quantum eigensolver.
\newblock {\em Quantum Science and Technology}, 11(1):015060, March 2026.

\bibitem{Utkarsh2023Chemistry}
Utkarsh Azad and Stepan Fomichev.
\newblock Pennylane quantum chemistry datasets.
\newblock \url{https://pennylane.ai/datasets/h3-plus-molecule}, 2023.

\bibitem{schollwockDensitymatrixRenormalizationGroup2011c}
Ulrich Schollwöck.
\newblock The density-matrix renormalization group in the age of matrix product states.
\newblock {\em Annals of Physics}, 326(1):96--192, January 2011.

\bibitem{carleoSolvingQuantumManybody2017a}
Giuseppe Carleo and Matthias Troyer.
\newblock Solving the quantum many-body problem with artificial neural networks.
\newblock {\em Science}, 355(6325):602--606, February 2017.

\bibitem{shiVariationalApproachManyBody2020}
Tao Shi, Eugene Demler, and J.~Ignacio Cirac.
\newblock Variational {Approach} for {Many}-{Body} {Systems} at {Finite} {Temperature}.
\newblock {\em Physical Review Letters}, 125(18):180602, October 2020.

\bibitem{luVariationalNeuralTensor2025}
Sirui Lu, Giacomo Giudice, and J.~Ignacio Cirac.
\newblock Variational neural and tensor network approximations of thermal states.
\newblock {\em Physical Review B}, 111(7):075102, February 2025.

\bibitem{mitaraiQuantumCircuitLearning2018}
K.~Mitarai, M.~Negoro, M.~Kitagawa, and K.~Fujii.
\newblock Quantum circuit learning.
\newblock {\em Physical Review A}, 98(3):032309, September 2018.

\bibitem{schuldCircuitcentricQuantumClassifiers2020}
Maria Schuld, Alex Bocharov, Krysta~M. Svore, and Nathan Wiebe.
\newblock Circuit-centric quantum classifiers.
\newblock {\em Physical Review A}, 101(3):032308, March 2020.

\bibitem{mariTransferLearningHybrid2020}
Andrea Mari, Thomas~R. Bromley, Josh Izaac, Maria Schuld, and Nathan Killoran.
\newblock Transfer learning in hybrid classical-quantum neural networks.
\newblock {\em Quantum}, 4:340, October 2020.

\end{thebibliography}

\clearpage
\newpage
\onecolumngrid

\setcounter{secnumdepth}{2}
\setcounter{section}{0}
\setcounter{subsection}{0}
\setcounter{equation}{0}
\setcounter{figure}{0}
\setcounter{table}{0}
\crefname{equation}{Eq.}{Eqs.}
\setcounter{section}{0}
\crefname{figure}{Fig.}{Figs.}
\crefname{section}{Appendix}{Appendices}

\renewcommand{\thefigure}{S\arabic{figure}}
\renewcommand{\theequation}{S\arabic{equation}}

\begin{center} 
    {\large \bf Supplemental Material: Generalized Nonlinear Imaginary-Time Evolution}
    \\[0.8em]

    {\normalsize
    Chenyu Shi$^{1,2}$,
    Hao Wang$^{1,2,*}$,
    and Jin-Fu Chen$^{2,3,\dagger}$
    }
    \\[0.8em]

    {\footnotesize
    $^{1}$Leiden Institute of Advanced Computer Science, Leiden University, Leiden, The Netherlands\\
    $^{2}$ $\langle aQa^L\rangle$ Applied Quantum Algorithms, Universiteit Leiden, The Netherlands\\
    $^{3}$Instituut-Lorentz, Universiteit Leiden, P.O. Box 9506, 2300 RA Leiden, The Netherlands
    }
\end{center}

\begingroup
\makeatletter
\long\def\@makefntext#1{\parindent=0pt\noindent #1}
\makeatother
\renewcommand{\thefootnote}{}
\footnotetext{
\begin{tabular}{@{}r@{\ }l@{}}
$^{*}$ & Corresponding author: h.wang@liacs.leidenuniv.nl\\
$^{\dagger}$ & Corresponding author: jinfuchen@lorentz.leidenuniv.nl
\end{tabular}
}
\endgroup

\section{State-dependent Hermitian Operator $H(\phi)$} \label{sec:H}

Suppose $\rho$ represents an $n$-qubit quantum system. Denote $N=2^n$ and $\mathcal{H}_{N}=\{X\in\mathbb{C}^{N\times N}:X=X^{\dagger}\}$ as the space of $N\times N$ Hermitian matrices, equipped with the Hilbert--Schmidt inner product $\langle X,Y\rangle_{HS}=\mathrm{Tr}(XY)$.

The cost function $\mathcal{E}:\mathcal{H}_N \rightarrow \mathbb{R}$ defines a function from the Hermitian matrix space to real values. The Hilbert--Schmidt gradient $\nabla_{\rho}^{(HS)}\mathcal{E}(\rho)$ is correspondingly defined as the unique element that satisfies:
\begin{equation}
    D\mathcal{E}(\rho)[X]=\mathrm{Tr}(\nabla_{\rho}^{(HS)}\mathcal{E}(\rho)\,X), \ \ \ \ \forall X \in \mathcal{H}_N.
    \label{eq:gradienth}
\end{equation}

We first prove that $\nabla_{\rho}^{(HS)}\mathcal{E}(\rho)$ is Hermitian. Since the Hermitian matrix space $\mathcal{H}_N$ is a finite-dimensional real vector space, the differential $D\mathcal{E}(\rho)$ defines a linear functional $X\mapsto D\mathcal{E}(\rho)[X]$. According to the Riesz representation theorem, there exists a unique $G\in \mathcal{H}_N$ whose inner product with $X$ satisfies:
\begin{equation}
    D\mathcal{E}(\rho)[X]=\mathrm{Tr}(G X), \ \ \ \ \forall X \in \mathcal{H}_N.
    \label{eq:reize}
\end{equation}
According to the definition in \cref{eq:gradienth}, this $G$ is the gradient $\nabla_{\rho}^{(HS)}\mathcal{E}(\rho)$. Hence, we conclude that $H(\phi)$ is always Hermitian.

Next we prove $H(\phi)\ket{\phi}=\ket{g(\phi)}$. Because $\rho(\phi)=\ket{\phi}\bra{\phi}$, we have $\delta_{\bra{\phi}} \rho(\phi) = \ket{\phi}\,\delta\bra{\phi}$. The chain rule gives
\begin{align}
    \frac{\delta}{\delta\bra{\phi}}E(\phi) &= \delta_{\bra{\phi}} \mathcal{E}(\rho(\phi)) \notag \\
    &= \frac{\delta \mathcal{E}(\rho)}{\delta \rho} \frac{\delta\rho}{\delta\bra{\phi}} \notag \\
    &= H(\phi)\ket{\phi},
\end{align}
where in the last equality we use $\nabla_{\rho}^{(\mathrm{HS})}\mathcal{E}(\rho)=\frac{\delta \mathcal{E}(\rho)}{\delta\rho}$. Hence, we have proven that $H(\phi)\ket{\phi}=\ket{g(\phi)}$.

In the case where the cost function is $E(\phi)=\braket{\phi|H|\phi}$, namely $\mathcal{E}(\rho)=\operatorname{Tr}(H\rho)$, we prove that $H(\phi)=H$, i.e., NITE reduces to ITE:
\begin{align}
    H(\phi) &= \nabla_{\rho}^{(\mathrm{HS})} \mathcal{E}(\rho)\Big|_{\rho=\ket{\phi}\bra{\phi}} \notag \\
    &= \frac{\delta \operatorname{Tr}(H\rho)}{\delta\rho} \notag \\
    &= H.
\end{align}

Hence, we have proven that the definition in \cref{eq:hermitian} is consistent with ITE.

\section{Local Convergence Rate}\label{sec:POT}
In this section, we first provide the proof of \cref{th:linear_converge}.
\begin{proof}
Fix $\mu\in(0,2\lambda_*)$, and choose $m$ such that
\begin{equation}
    \frac{\mu}{2}<m<\lambda_* .
\end{equation}
By continuity of the Fubini--Study Hessian, after shrinking to a sufficiently
small geodesically convex neighborhood $U$ of $\boldsymbol{\phi}_*$, we have
\begin{equation}
    \operatorname{Hess}_{\mathrm{FS}}E(\boldsymbol{\phi})[\xi,\xi]
    \ge
    m\|\xi\|_{\mathrm{FS}}^{2},
    \qquad
    \boldsymbol{\phi}\in U,\quad
    \xi\in T_{\boldsymbol{\phi}}\mathcal{M}_{K}.
\end{equation}
Hence $E$ is $m$-strongly geodesically convex on $U$. Since
$\boldsymbol{\phi}_*$ is a critical point, this implies the local
Polyak--{\L}ojasiewicz inequality
\begin{equation}
    \left\|
        \nabla^{(\mathrm{FS})}E(\boldsymbol{\phi})
    \right\|_{\mathrm{FS}}^{2}
    \ge
    2m
    \left(
        E(\boldsymbol{\phi})-E(\boldsymbol{\phi}_*)
    \right),
    \qquad
    \boldsymbol{\phi}\in U .
\label{eq:normEE}
\end{equation}

Choose a smaller neighborhood $V\subset U$ of $\boldsymbol{\phi}_*$ such
that every NITE trajectory starting from $V$ remains in $U$. This can be
done by taking $V$ to be a sufficiently small sublevel neighborhood of
$E$, since $E$ decreases along the NITE flow and
$\boldsymbol{\phi}_*$ is a strict local minimum.

Along the NITE flow,
\begin{equation}
    \frac{d}{dt}E(\boldsymbol{\phi}(t))
    =
    -
    \left\|
        \nabla^{(\mathrm{FS})}E(\boldsymbol{\phi}(t))
    \right\|_{\mathrm{FS}}^{2}.
    \label{eq:dtEE}
\end{equation}
Therefore, for every initial condition $\boldsymbol{\phi}(0)\in V$, substituting \cref{eq:normEE} into \cref{eq:dtEE}, we obtain:
\begin{equation}
    \frac{d}{dt}
    \left(
        E(\boldsymbol{\phi}(t))-E(\boldsymbol{\phi}_*)
    \right)
    \le
    -2m
    \left(
        E(\boldsymbol{\phi}(t))-E(\boldsymbol{\phi}_*)
    \right).
\end{equation}
By Gronwall's inequality,
\begin{equation}
    E(\boldsymbol{\phi}(t))-E(\boldsymbol{\phi}_*)
    \le
    e^{-2mt}
    \left(
        E(\boldsymbol{\phi}(0))-E(\boldsymbol{\phi}_*)
    \right).
\end{equation}
Since $2m>\mu$, the desired estimate follows.
\end{proof}

Then we provide the details of remarks. 
\begin{remark}
For the standard imaginary-time evolution with the expectation-value cost,
the positive-definiteness assumption in \cref{th:linear_converge} reduces
to the usual non-degenerate spectral gap condition.
\end{remark}

\begin{proof}
Consider the single-state case $K=1$, so that
$\mathcal M_1=\mathbb{CP}^{d-1}$, and let
\begin{equation}
    E(\psi)=\langle \psi,H\psi\rangle ,
    \qquad \|\psi\|=1 .
\end{equation}
Let
\begin{equation}
    H\ket{u_j}=E_j\ket{u_j},
    \qquad
    E_0<E_1\le E_2\le\cdots ,
\end{equation}
and suppose that the ground state $\ket{u_0}$ is non-degenerate. Around
$\ket{u_0}$, use the affine chart
\begin{equation}
    \ket{\psi(z)}
    =
    \frac{
        \ket{u_0}+\sum_{j\ge1}z_j\ket{u_j}
    }{
        \left(1+\sum_{j\ge1}|z_j|^2\right)^{1/2}
    } .
\end{equation}
Then
\begin{equation}
    E(\psi(z))
    =
    \frac{
        E_0+\sum_{j\ge1}E_j|z_j|^2
    }{
        1+\sum_{j\ge1}|z_j|^2
    }
    =
    E_0
    +
    \sum_{j\ge1}(E_j-E_0)|z_j|^2
    +
    O(\|z\|^4).
\end{equation}
Thus the quadratic part of $E-E_0$ in the Fubini--Study normal
coordinates is diagonal, with eigenvalues
$    E_j-E_0$, $j\ge1,$
under our normalization of the Fubini--Study Hessian. Therefore, we obtain
\begin{equation}
    \lambda_{\min}
    \left(
        \operatorname{Hess}_{\mathrm{FS}}E(u_0)
    \right)
    =
    E_1-E_0 .
\end{equation}
Hence the Hessian is positive definite if and only if the ground state is
non-degenerate and separated from the rest of the spectrum by a positive
gap. In this case, the constant $\lambda_*$ in
\cref{th:linear_converge} is exactly the spectral gap
\begin{equation}
    \lambda_* = E_1-E_0 .
\end{equation}
\end{proof}

We next formulate the theorem for the general case where the optimum is not a single point but a submanifold.

\begin{theorem}[Exponential convergence to optimal submanifold]
Let $\mathcal M_K=(\mathbb{CP}^{d-1})^K$ be equipped with the product
Fubini--Study metric, and let $E:\mathcal M_K\to\mathbb R$ be $C^2$.
Suppose that $\boldsymbol\phi_*$ belongs to a smooth local manifold
$\mathcal N$ of minimizers of $E$, and that $\ker\operatorname{Hess}_{\mathrm{FS}}E(\boldsymbol\phi_*)=T_{\boldsymbol\phi_*}\mathcal N$. Let $\lambda_*=\lambda_{\min}\left(\operatorname{Hess}_{\mathrm{FS}}E(\boldsymbol\phi_*)\big|_{(T_{\boldsymbol\phi_*}\mathcal N)^\perp}\right)>0$.  Then, for any $\mu\in(0,2\lambda_*)$, there exists a neighborhood $V$ of $\boldsymbol\phi_*$ such that, for every initial point $\boldsymbol\phi(0)\in V$ whose NITE trajectory converges to $\phi_*$, the evolution satisfies, for all $t\geq0$,
\begin{equation*}
    E(\boldsymbol\phi(t))-E(\boldsymbol\phi_*)
    \le
    e^{-\mu t}
    \left(
        E(\boldsymbol\phi(0))-E(\boldsymbol\phi_*)
    \right),    
\end{equation*}

\end{theorem}
\begin{proof}
Fix $\mu\in(0,2\lambda_*)$, and choose $m$ such that
\begin{equation}
    \frac{\mu}{2}<m<\lambda_* .
\end{equation}
Since
\begin{equation}
    \ker\operatorname{Hess}_{\mathrm{FS}}E(\boldsymbol\phi_*)
    =
    T_{\boldsymbol\phi_*}\mathcal N
\end{equation}
and
\begin{equation}
    \lambda_*
    =
    \lambda_{\min}
    \left(
    \operatorname{Hess}_{\mathrm{FS}}E(\boldsymbol\phi_*)
    \big|_{(T_{\boldsymbol\phi_*}\mathcal N)^\perp}
    \right)>0,
\end{equation}
the Hessian is positive definite in directions normal to
$\mathcal N$ at $\boldsymbol\phi_*$. By continuity of the
Fubini--Study Hessian, after shrinking to a sufficiently small tubular
neighbourhood $U$ of $\mathcal N$ near $\boldsymbol\phi_*$, the
normal Hessian remains bounded below by $m$.

Equivalently, in local coordinates $(u,v)$ adapted to $\mathcal N$,
where $u$ is tangent to $\mathcal N$ and $v$ is normal to
$\mathcal N$, we have
\begin{equation}
    E(u,v)-E(\boldsymbol\phi_*)
    =
    \frac12
    \operatorname{Hess}_{\mathrm{FS}}E(\boldsymbol\phi_*)[v,v]
    +
    o(\|v\|_{\mathrm{FS}}^{2}),
\end{equation}
and
\begin{equation}
    \nabla^{(\mathrm{FS})}E(u,v)
    =
    \operatorname{Hess}_{\mathrm{FS}}E(\boldsymbol\phi_*)v
    +
    o(\|v\|_{\mathrm{FS}}).
\end{equation}
Therefore, after possibly shrinking $U$, $E$ satisfies the local
Polyak--{\L}ojasiewicz inequality modulo gauge:
\begin{equation}
    \left\|
        \nabla^{(\mathrm{FS})}E(\boldsymbol\phi)
    \right\|_{\mathrm{FS}}^{2}
    \ge
    2m
    \left(
        E(\boldsymbol\phi)-E(\boldsymbol\phi_*)
    \right),
    \qquad
    \boldsymbol\phi\in U .
\end{equation}

Choose a smaller neighborhood $V\subset U$ of $\boldsymbol\phi_*$
such that every NITE trajectory starting from $V$ and converging to
$\boldsymbol\phi_*$ remains in $U$. This can be done by taking $V$
sufficiently small, since $E$ decreases along the NITE flow and
$\boldsymbol\phi_*$ is a local minimizer modulo the gauge directions.

Along the NITE flow,
\begin{equation}
    \frac{d}{dt}E(\boldsymbol\phi(t))
    =
    -
    \left\|
        \nabla^{(\mathrm{FS})}E(\boldsymbol\phi(t))
    \right\|_{\mathrm{FS}}^{2}.
\end{equation}
Hence, for every initial condition $\boldsymbol\phi(0)\in V$ whose
trajectory converges to $\boldsymbol\phi_*$,
\begin{equation}
    \frac{d}{dt}
    \left(
        E(\boldsymbol\phi(t))-E(\boldsymbol\phi_*)
    \right)
    \le
    -2m
    \left(
        E(\boldsymbol\phi(t))-E(\boldsymbol\phi_*)
    \right).
\end{equation}
Gronwall's inequality gives
\begin{equation}
    E(\boldsymbol\phi(t))-E(\boldsymbol\phi_*)
    \le
    e^{-2mt}
    \left(
        E(\boldsymbol\phi(0))-E(\boldsymbol\phi_*)
    \right).
\end{equation}
Since $2m>\mu$, the desired estimate follows.
\end{proof}

\section{NITE of a single state to QNGD}\label{sec:NITE2QNGD}
For the variational state $\ket{\phi(\theta)}$, \cref{eq:NITE} takes the following form by the chain rule
\begin{equation}
    \frac{\partial \ket{\phi}}{\partial \theta_i}\dot{\theta_i}
    = -(1-\ket{\phi}\bra{\phi})H(\phi)\ket{\phi}.
    \label{eq:varnite1}
\end{equation}
The variational principle then leads to the following least-squares problem
\begin{align}
    \dot{\theta_i}
    &= \mathop{\arg\min}\limits_{\dot{\theta_i}}
    \left\|
    \frac{\partial \ket{\phi}}{\partial \theta_i}\dot{\theta_i}
    +(1-\ket{\phi}\bra{\phi})H(\phi)\ket{\phi}
    \right\|_{\mathrm{FS}}^{2} \notag \\
    &= \mathop{\arg\min}\limits_{\dot{\theta_i}}
    \left\|
    (1-\ket{\phi}\bra{\phi})\frac{\partial \ket{\phi}}{\partial \theta_i}\dot{\theta_i}
    +(1-\ket{\phi}\bra{\phi})H(\phi)\ket{\phi}
    \right\|_{2}^{2} .
\end{align}
Expanding the above least-squares problem, we obtain an expression of the form $\dot{\theta}^TA\dot{\theta}+b\dot{\theta}+c$, where the matrix $A$ is

\begin{align}
    A_{ij}  &=\bra{\partial_{\theta_i}\phi}(1-\ket{\phi}\bra{\phi})^2\ket{\partial_{\theta_j}\phi} \notag \\
    &=\braket{\partial_{\theta_i}\phi|\partial_{\theta_j}\phi}-\braket{\partial_{\theta_i}\phi|\phi}\braket{\phi|\partial_{\theta_j}\phi}.
\end{align}

Note that $\sum_{ij}A_{ij}\dot{\theta}_i\dot{\theta}_j$ is a real number, so using $\operatorname{Re}[A]$ or $A$ in the summation gives the same result: $\sum_{ij}A_{ij}\dot{\theta}_i\dot{\theta}_j=\sum_{ij}\operatorname{Re}[A]_{ij}\dot{\theta}_i\dot{\theta}_j$. Hence, we recover the Fubini--Study metric tensor used in QNGD.

For the term $b$, we have
\begin{align}
    b_i &= \bra{\partial_{\theta_i}\phi}(1-\ket{\phi}\bra{\phi})H(\phi)\ket{\phi}+\bra{\phi}H(\phi) (1-\ket{\phi}\bra{\phi})\ket{\partial_{\theta_i}\phi} \notag \\
    &=\bra{\partial_{\theta_i}\phi}H(\phi)\ket{\phi}+\bra{\phi}H(\phi)\ket{\partial_{\theta_i}\phi}-\braket{\phi|H(\phi)|\phi}(\braket{\partial_{\theta_i}\phi|\phi}+\braket{\phi|\partial_{\theta_i}\phi}) \notag \\
    &=\bra{\partial_{\theta_i}\phi}H(\phi)\ket{\phi}+\bra{\phi}H(\phi)\ket{\partial_{\theta_i}\phi}.
    \label{eq:bterm}
\end{align}
Taking the derivative of the cost function with respect to the parameters $\theta_i$, we have
\begin{align}
    \frac{\partial}{\partial \theta_i}E(\phi(\theta)) &= \frac{\partial \bra{\phi}}{\partial \theta_i} \frac{\partial E(\phi)}{\partial \bra{\phi}}+\frac{\partial E(\phi)}{\partial \ket{\phi}}\frac{\partial \ket{\phi}}{\partial \theta_i} \notag \\
    &= \bra{\partial_{\theta_i}\phi}H(\phi)\ket{\phi}+\bra{\phi}H(\phi)\ket{\partial_{\theta_i}\phi} \notag \\
    &= b_i.
\end{align}
Hence, we have shown that $b=\nabla_{\theta} E(\phi(\theta))$. By the KKT conditions, we have $\dot{\theta}=-A^{-1}b$. Using a discrete time step as the learning rate, we exactly recover the QNGD update step in \cref{eq:qngd}.

\section{NITE of multiple states to generalized QNGD} \label{sec:NITEmulti2QNGD}
For each variational state $\ket{\phi_k(\theta)}$, \cref{eq:NNITEmulti} takes the following form by the chain rule:
\begin{equation}
    \frac{\partial \ket{\phi_k}}{\partial \theta_i}\dot{\theta_i}
    = -(1-\ket{\phi_k}\bra{\phi_k})H_k(\boldsymbol{\phi})\ket{\phi_k}.
    \label{eq:varnite2}
\end{equation}
The variational principle then leads to the following least-squares problem for the $K$ differential equations:
\begin{align}
\dot{\theta_i}
&= \mathop{\arg\min}\limits_{\dot{\theta_i}}
\sum_{k=1}^{K}
\left\|
(1-\ket{\phi_k}\bra{\phi_k})\frac{\partial \ket{\phi_k}}{\partial \theta_i}\dot{\theta_i}
+(1-\ket{\phi_k}\bra{\phi_k})H(\boldsymbol{\phi})\ket{\phi_k}
\right\|_{2}^{2}.
\end{align}

Similar to the previous section, the above least-squares problem can be expanded in the form $\dot{\theta}^T A \dot{\theta} + b \dot{\theta} + c$. The matrix $A$ is now given by:
\begin{align}
A_{ij}&=\sum_k \bra{\partial_{\theta_i}\phi_k}(1-\ket{\phi_k}\bra{\phi_k})^2\ket{\partial_{\theta_j}\phi_k} \notag \\
&=\sum_k\braket{\partial_{\theta_i}\phi_k|\partial_{\theta_j}\phi_k}-\braket{\partial_{\theta_i}\phi_k|\phi_k}\braket{\phi_k|\partial_{\theta_j}\phi_k} \notag \\
&= \sum_k [G_k]_{ij}.
\end{align}
Compared to the NITE of a single state, the matrix $A=\sum_k G_k$ for the NITE of multiple states is the sum of the Fubini--Study metric tensors of all pure states in $\boldsymbol{\phi}$.

For the term $b$, according to the proof in \cref{eq:bterm}, we similarly have $b_i=\sum_k \Big(\bra{\partial_{\theta_i}\phi_k}H(\boldsymbol{\phi})\ket{\phi_k}+\bra{\phi_k}H(\boldsymbol{\phi})\ket{\partial_{\theta_i}\phi_k}\Big)$. Taking the derivative of the cost function with respect to the parameters $\theta_i$, we have
\begin{align}
    \frac{\partial}{\partial \theta_i}E(\boldsymbol{\phi}(\theta)) &= \sum_k \frac{\partial \bra{\phi_k}}{\partial \theta_i} \frac{\partial E(\boldsymbol{\phi})}{\partial \bra{\phi_k}}+\frac{\partial E(\boldsymbol{\phi})}{\partial \ket{\phi_k}}\frac{\partial \ket{\phi_k}}{\partial \theta_i} \notag \\
    &= \sum_k \bra{\partial_{\theta_i}\phi_k}H(\boldsymbol{\phi})\ket{\phi_k}+\bra{\phi_k}H(\boldsymbol{\phi})\ket{\partial_{\theta_i}\phi_k} \notag \\
    &= \sum_k [b_k]_i.
\end{align}
Therefore, we have shown that $b=\nabla_{\theta} E(\phi(\theta))$. By the KKT conditions, we have $\dot{\theta}=-A^{-1}b$. Using a discrete time step as the learning rate, we exactly recover the QNGD update step in \cref{eq:qngdmulti} of the main text.

\section{Supplement Numerics} \label{sec:SN}
To avoid the limitations of the ansatz, in this section we present supplementary numerical results obtained using full-space parameterization. The full-space parameterization directly treats the state vector as a $2^n$-dimensional complex vector. Note that this full-space parameterization requires $2^{n+1}$ real parameters; hence, it is not efficiently implementable in practice compared to the hardware-efficient ansatz used in the main text. The results shown here are intended solely to show that the limitation in convergence arises from the expressivity of the ansatz, rather than from the method itself. In addition, the RK45 method is used to solve the differential equations directly, instead of the QNGD approximation as a first-order discretization. For a fair comparison, we also use RK45 to solve the differential equation corresponding to gradient descent.

\cref{fig:variance_log} shows the learning curves for minimizing the variance. The two methods perform similarly at the initial stage. The gradient descent finally stucks around $10^{-3}$ precision, while NITE still exhibits a linear convergence rate. 

\begin{figure}
\centering
\includegraphics[width=0.4\linewidth]{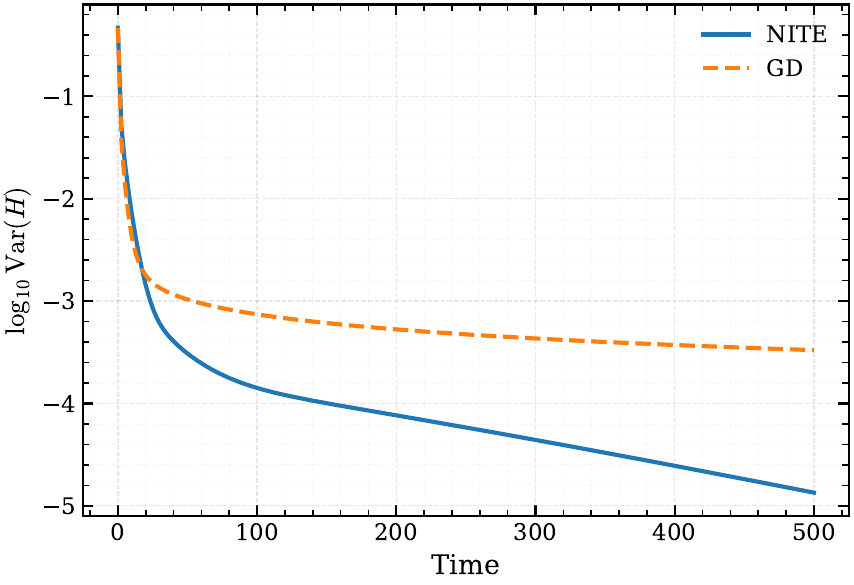}
\caption{Learning curves for minimizing variance on a logarithmic scale with full-space parameterization.}
\label{fig:variance_log}
\end{figure}

\begin{figure}
\centering
\includegraphics[width=0.4\linewidth]{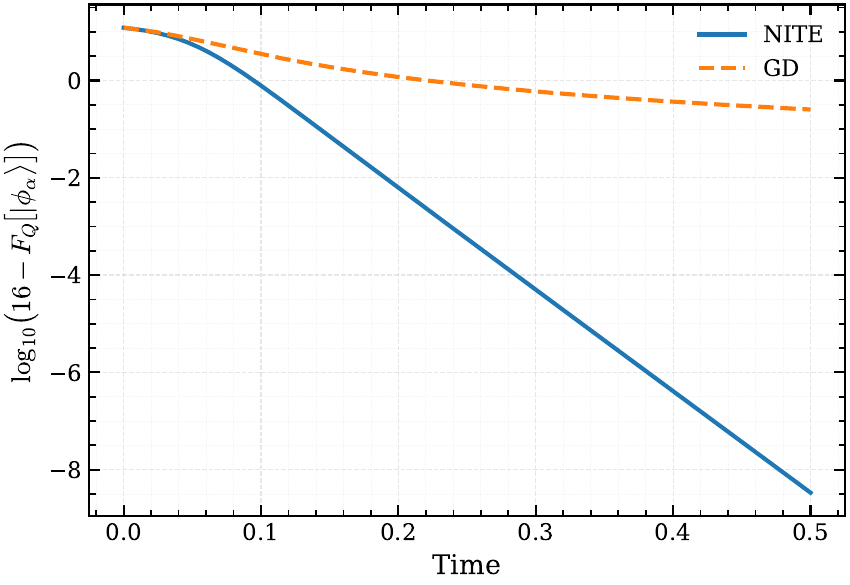}
\caption{Learning curves for maximizing QFI on a logarithmic scale with full-space parameterization.}
\label{fig:qfilog}
\end{figure}

\begin{figure}[htbp]
    \centering
    \subfloat[]{
        \includegraphics[width=0.47\linewidth]{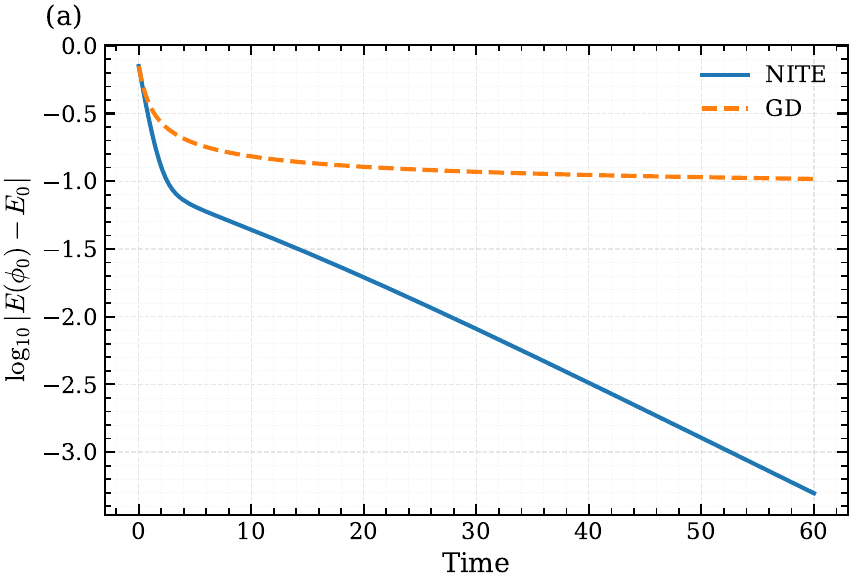}
        \label{fig:fig11}
    }
    \hfill
    \subfloat[]{
        \includegraphics[width=0.47\linewidth]{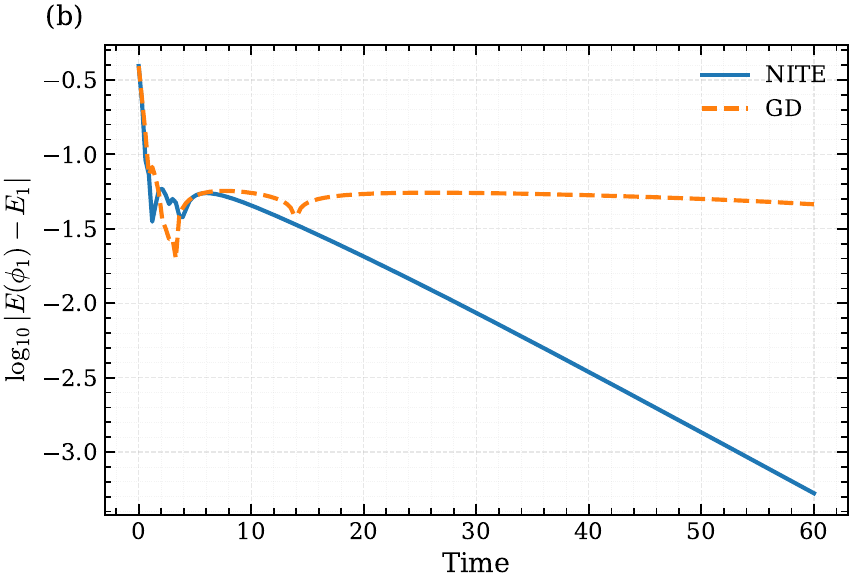}
        \label{fig:fig22}
    }
    \caption{
    Log-scale learning curves for preparing (a) the ground state and
    (b) the first excited state.
    }
    \label{fig:excited_log}
\end{figure}

\cref{fig:qfilog} shows the learning curves for preparing the probe state to maximize the QFI. NITE achieves a linear convergence rate, while gradient descent exhibits slow sublinear convergence behavior as the optimization progresses.

\cref{fig:excited_log} presents the learning curves for preparing the ground and first excited states. The left panel corresponds to the ground state, and the right panel to the first excited state. For both states, gradient descent stalls at limited precision, while NITE continues to exhibit a linear convergence rate.

The supplementary numerical results in this section present a performance comparison in the regime where the state parameterization is sufficiently expressive, thereby avoiding limitations due to ansatz expressibility and enabling a clearer comparison between the two approaches as optimization methods. The results indicate that NITE has advantages over gradient descent in all subroutine tasks.

\begin{figure}
\centering
\includegraphics[width=0.94\linewidth]{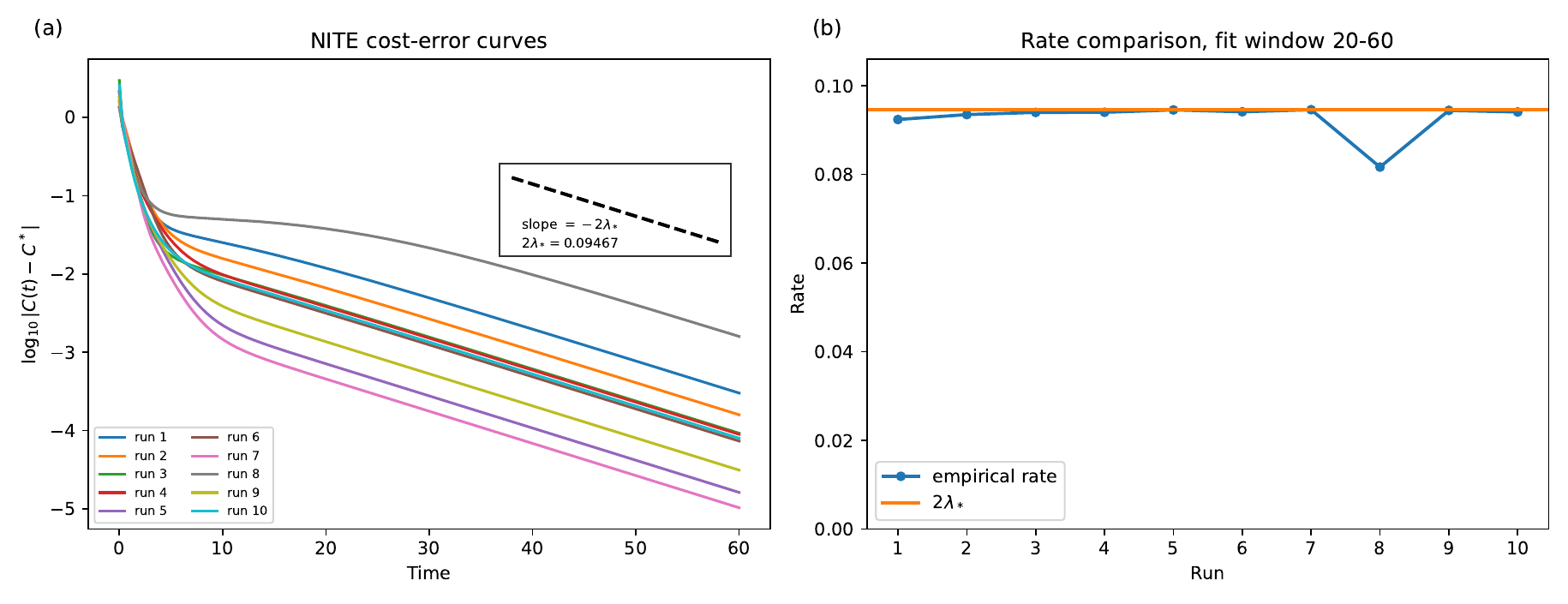}
\caption{
Comparison between the numerical NITE convergence rate and the Hessian-predicted local rate.
(a) Logarithmic cost-error trajectories for ten independent NITE runs. The dashed black line indicates the theoretical local slope predicted by $2\lambda_*$.
(b) Numerical convergence rates extracted from a linear fit of $\log |C(t)-C^*|$ over the fitting window $t\in[20,60]$, compared with the Hessian prediction $2\lambda_*$.
}
\label{fig:raterate}
\end{figure}

\end{document}